\documentclass[a4paper,USenglish,cleveref,autoref,thm-restate,pdfa,
    nolineno, 
]{socg-lipics-v2021}

\hideLIPIcs

\title{\texorpdfstring{$c$}{c}-Packedness versus \texorpdfstring{$\lambda$}{λ}-Low-Density in Geometric Graphs: Theory and Practice}

\author{Gregor Diatzko}{University of Konstanz, Germany}{gregor.diatzko@uni-konstanz.de}{https://orcid.org/0000-0002-0904-4910}{}
\author{Félix Lasseux}{University of Konstanz, Germany \and ENSEIRB, Bordeaux, France}{felix.lasseux@labri.fr}{https://orcid.org/0009-0003-8312-4702}{}

\author{Sabine Storandt}{University of Konstanz, Germany}{sabine.storandt@uni-konstanz.de}{https://orcid.org/0000-0001-5411-3834}{}

\authorrunning{G. Diatzko, F. Lasseux and S. Storandt} 

\Copyright{Gregor Diatzko, Félix Lasseux and Sabine Storandt}

\ccsdesc[500]{Theory of computation~Graph algorithms analysis}

\keywords{Geometric Graph, Balanced Separator, Tree Decomposition, Distance Oracle}

\relatedversion{This is the full version of an article accepted at IPEC 2026.} 

\supplementdetails[linktext={}, subcategory={Source Code}, swhid={swh:1:rev:a646306525afd52f5ce065881fa9e36ba6405d7a}]{Software}{https://gitlab.inf.uni-konstanz.de/felix.lasseux/c-packed_edos} 
\supplementdetails[
    linktext={https://gitlab.inf.uni-konstanz.de/ag-storandt/packedness-vs-density},
    subcategory={Source Code},
    swhid={swh:1:rev:01098294157b00c5692ed933fa6d6b7d3c449aae}]
    {Software}
    {https://gitlab.inf.uni-konstanz.de/ag-storandt/packedness-vs-density}

\EventEditors{Tuukka Korhonen and Michael Lampis}
\EventNoEds{2}
\EventLongTitle{21th International Symposium on Parameterized and Exact Computation (IPEC 2026)}
\EventShortTitle{IPEC 2026}
\EventAcronym{IPEC}
\EventYear{2026}
\EventDate{September 2--4, 2026}
\EventLocation{L’Aquila, Italy}
\EventLogo{}
\SeriesVolume{398}
\ArticleNo{27}

\AddToHook{cmd/appendix/before}{%
  \crefalias{section}{appendix}%
}
\usepackage{siunitx}

\usepackage{todonotes}
\def\O{\mathcal O}
\def\R{\mathbb R}
\def\tw{{\texttt{tw}}}
\def\reed{\textsc{reed}}
\let\epsilon=\varepsilon
\newcommand{\blackdiamond}{%
  \mathord{%
    \sbox0{$\blacksquare$}%
    \resizebox{!}{1.1\ht0}{%
      \rotatebox[origin=c]{45}{$\blacksquare$}%
    }%
  }%
}

\definecolor{kn_bordeaux65}{RGB}{165,77,105}
\definecolor{kn_peach}{RGB}{254,160,144}
\definecolor{kn_seegruen}{RGB}{10,144,134}
\definecolor{kn_seeblau}{RGB}{0,169,224}

\begin{document}

\maketitle

\begin{abstract}
When designing algorithms for geometric graphs, exploiting structural parameters
can lead to significantly improved bounds.
Two prominent parameters in this context are $c$-packedness and
$\lambda$-low density, both of which locally restrict graph complexity.
Parameterized algorithms based on these parameters have been developed for
computing well-separated pair decompositions, balanced separators, as well as
distance oracles.
Despite this theoretical progress, the practical applicability of algorithms
parameterized by $c$ or $\lambda$ remains unclear.
While $c$-packed and $\lambda$-low-density graphs have been proposed as
realistic models for road networks, the actual parameter values of large
real-world instances have so far remained unknown, and existing theoretical
guarantees are partially too loose for practical usage.
In this paper we first devise scalable implementations for the approximate
computation of $c$ and the exact computation of $\lambda$.
Our experimental evaluation on road networks with millions of edges reveals a
significant gap between the two parameters.
On the theoretical side we prove that $c \in \O(\lambda \sqrt n)$
which complements the known result that $\lambda \in \O(c)$.
Furthermore we present improved parameterized algorithms for balanced separator
computation that reduce the separator size in theory and practice.
We also show how to compute a tree decomposition with a width linear in the
respective parameterized balanced separator size in polynomial time.
This structural result yields a variety of new algorithmic consequences.
Among them is an exact distance oracle with query time $\O(c)$ for $c$-packed
graphs after polynomial-time preprocessing, which improves upon the previous
$\O(c \log n)$ bound.
Our experiments demonstrate that the proposed techniques efficiently produce
small balanced separator and enable the construction of concise exact distance oracles on large
road networks.
\end{abstract}

\section{Introduction}
Identifying realistic graph classes that model road networks is an ongoing
endeavor, with the goal to explain empirical observations---such as the
existence of small balanced separators---and to develop efficient algorithms and
data structures with provable bounds.
Parameterized graph classes have been of great interest in this context.
Most of the parameters considered so far, as treewidth
\tw~\cite{bauer2016search}, highway dimension {\tt hd}~\cite{abraham2016highway}
or skeleton dimension {\tt k}~\cite{kosowski2017beyond},
are based on the graph topology and the shortest path structure.
For ${\tt p} \in \{\tw, {\tt hd}, {\tt k}\}$ exact distance oracles with query times in
$\O({\tt p} \log n)$ can be constructed on a given graph $G(V,E)$ with $n$
nodes.
The respective oracle construction algorithms are impractical though, as even
approximate variants have running times that are either exponential in the
parameter ${\tt p}$ or (super-)quadratic in $n$.
Throughput the paper, we call a parametrized algorithm \emph{practical} if it
has a running time of the form $poly({\tt p} ) \cdot \tilde{\mathcal{O}}(n)$, as
this allows it to scale to large graphs even when the parameter value is
moderate.
To achieve such running times, typically  heuristics are  used in implementations of parameterized distance oracles 
instead of exact or approximate algorithms
\cite{gottesburen2019faster,delling2014hub,blum2021sublinear}, potentially
forfeiting their respective theoretical guarantees.

However, recently new parameterized graph classes for road networks have been
studied that solely rely on the geometry of the graph embedding.
Two particularly successful notions are $c$-packed graphs and $\lambda$-low
density graphs.
Informally, an embedded graph is $c$-packed if the total length of edges
contained in any ball of radius $r$ is at most $cr$.
Likewise, a graph has $\lambda$-low density if every ball of radius $r$ intersects
at most $\lambda$ edges of length at~least~$r$.
Thus, a small parameter values implies a certain notion of sparsity in both
definitions.
Practical algorithms for a variety of problems have been proposed under such
sparsity assumptions.
For $c$-packed graphs, this includes the computation of a WSPD of size
$\mathcal{O}(c^3 n)$ in time $\mathcal{O}(c^3 n \log n)$, a balanced separator
of size $\mathcal{O}(c)$ in time $\mathcal{O}(c^2n)$, and an exact distance
oracle with a preprocessing time in $\mathcal{O}(c^2n \log^2 n)$ and a query
time in $\mathcal{O}(c \log n)$ \cite{deryckere2025wspd}.
For $\lambda$-low-density graphs, a WSPD can be constructed in time
$\mathcal{O}(\lambda^2 n \log n)$ \cite{gudmundsson2026well}, and a balanced
separator of size $\mathcal{O}(\lambda \sqrt n )$ can be computed in expected
time $\mathcal{O}(\lambda n)$ \cite{le2024greedy}.

To assess the usefulness of these bounds, the question arises what the values of  $c$ and $\lambda$ are in large real-world networks. 
In this paper, we engineer algorithms to extract $c$ and $\lambda$, enabling us to compute their (approximate) values for road networks with millions of edges for the first time. Furthermore, we investigate the practical applicability of  algorithms parameterized by $c$ or $\lambda$. We identify several theoretical bottlenecks and propose improved and novel algorithms to overcome them. In particular, we present the first exact distance oracle with a practical preprocessing time and  $\mathcal{O}(c)$ query time. Unlike previous parameterized distance oracles, our implementation adheres to the theoretical guarantees while also scaling to large networks and exhibiting excellent query times in practice.

\subsection{Related Work}
The notion of $c$-packedness was originally proposed to study polygonal curves \cite{driemel2010approximating}. It was shown in the same paper that computing the  Fréchet distance  between $c$-packed curves only takes 
near-linear time. In \cite{gudmundsson2015fast}, the concept was transferred to graphs and studied in the context of map matching. A plethora of other results
for $c$-packed curves or graphs have been established in the meantime, see e.g. \cite{bringmann2017improved,gudmundsson2024map,deryckere2025wspd,conradi2026computing,gudmundsson2023computing}.
For the computation of $c$, an exact algorithm with a running time in $\mathcal{O}(m^3)$ is known, where $m$ is the number of edges,  if axis-aligned squares are used in the definition instead of balls \cite{gudmundsson2013algorithms}. Furthermore, approximation algorithms have been investigated, including a $2$-approximation with a running time in $\mathcal{O}(m^2 \log m)$ \cite{gudmundsson2023approximatingc}, a $(6+\varepsilon)$-approximation with a running time in $\tilde{\mathcal{O}}((m/\varepsilon^3)^{4/3})$ \cite{narasimhan2007geometric}, and a recent  $42$-approximation with an expected running time in $\mathcal{O}(m\log^2 m)$ and high success probability \cite{harpeled2025how}. In \cite{aghamolaei2023sampling}, the authors propose the first exact algorithm for $c$-packedness defined on balls with a running time in $\mathcal{O}(m^5)$. Furthermore, they also improve the running time of the $2$-approximation to $\mathcal{O}(m^2)$ and the guarantee of the $(6+\varepsilon)$-approximation to $(4+\varepsilon)$. To the best of our knowledge, only \cite{gudmundsson2023approximatingc} also reports empirical results for  $c$, in particular for the $2$-approximation. They evaluate their implementation on trajectory sets with small to moderate size. Oftentimes small $c$ values are reported. However no general graphs were considered and no running times were provided.

The study of $\lambda$-low-density was initiated in
\cite{vanderstappen1993complexity} in the context of motion planning amid
obstacles. Later, it was shown to also be useful for map matching
\cite{chen2011approximate} as well as the design of approximation algorithms
\cite{harpeled2017approximation} and data structures \cite{gudmundsson2026well}.
An exact algorithm to compute $\lambda$ with a running time in $\mathcal{O}(n \log^3 n + \lambda n \log^2 n + \lambda^2n)$ was proposed in \cite{deberg2002realistic} and a $3$-approximation with an improved running time of  $\mathcal{O}(n \log n + \lambda n)$ in \cite{gudmundsson2023approximatinglambda}. 
The exact algorithm was tested in~\cite{deberg2002realistic} on triangulations
with up to 5000 points, and in~\cite{chen2011approximate} on small city road
networks. In~\cite{gudmundsson2023approximatinglambda} a $4$-approximation for
$\lambda$ (a slight simplification of the $3$-approximation) was implemented and
evaluated on a trajectory set with the largest containing roughly $n=\num{65000}$
points. But again no running times were reported in any of these papers and thus
scalability to  general graphs, especially with million of edges, remained
unclear.

\subsection{Contribution}
We present  novel theoretical and practical results for algorithms parameterized by $c$ or $\lambda$. 
\begin{itemize}
    \item \textit{Parameter relationship.} We prove that $c \in \mathcal{O}(\lambda \sqrt n)$ and that this bound is tight. The upper bound allows to transform algorithms parameterized by $c$ into algorithms parameterized by $\lambda$ and complements the known result $ \lambda \in \mathcal{O}(c)$ \cite{driemel2010approximating}. 
    \item  \textit{Parameter computation.} We present the first implementation for the exact computation of $\lambda$ that scales to graphs with millions of edges as well efficient approximations for $c$.  Our experiments show that both parameters remain moderate even on very large road networks, with  $\lambda$ being  consistently much smaller than $c$. 
    \item \textit{Balanced separators.}
    It was shown in \cite{deryckere2025wspd} that $c$-packed graphs admit
    balanced separators of size $\mathcal{O}(c)$ and that such separators can be
    computed efficiently. However, we observe that the hidden constant in the
    size bound  is prohibitively large in practice. We therefore develop
    techniques that substantially reduce the separator size without increasing
    the asymptotic running time. A key ingredient is an improved PTAS for the
    smallest enclosing square problem, which may be of independent interest. We
    also present an engineered separator algorithm and show that it consistently
    computes smaller balanced separators on road networks than KaHIP, a
    state-of-the-art graph partitioning framework \cite{sanders2013think}.
    \item \textit{Tree decomposition.}
    We show that a tree decomposition of width $\mathcal{O}(c)$ can be computed
    in time~$\mathcal{O}(c^2 n \log n)$ time. While it was previously known that the
    treewidth of $c$-packed graphs is bounded by $\mathcal{O}(c)$, existing
    algorithms for constructing such a decomposition were assumed to require
    exponential time \cite{deryckere2025wspd}. Our efficient construction
    therefore enables, for the first time, the practical use of treewidth-based
    algorithms on $c$-packed graphs after only polynomial-time preprocessing. As
    one application, combining our algorithm with the exact distance oracle of
    \cite{conrado2024faster} yields a construction time $\mathcal{O}(c^3 n \log
    n)$ and distance query time $\mathcal{O}(c)$. This improves the previous
    best query time of $\mathcal{O}(c \log n)$ for $c$-packed graphs. We also
    provide a careful implementation of our tree decomposition algorithm and
    show that it produces small bags on large real-world road networks, enabling
    the construction of highly efficient exact distance oracles in practice.
\end{itemize}


\section{Parameter Definitions and Relationships}
In the following, we consider simple undirected graphs embedded in the plane. Each vertex is associated with a point in the plane, and every edge is assigned its Euclidean length.

The definitions of the parameters considered in this paper appear in the literature using either balls or axis-aligned squares. Since every square can be covered by a constant number of balls of the same radius and every ball can be covered by a constant number of axis-aligned squares of the same radius,
the two variants are equivalent up to a constant factor. Throughout the paper, the radius of a square is defined as  half its side length.
\begin{definition}[$\lambda$-low density]
\label{def:low-density}
A graph $G$ has $\lambda$-low density
if for each square (ball) $S$ of radius $r$,
there are at most $\lambda$ edges of $G$
that have length greater than $r$ and intersect $S$.
\end{definition}
\begin{definition}[$c$-packedness]
\label{def:packedness}
A graph $G$ is $c$-packed
if for each square (ball) $S$ of radius $r$,
the total length of the edges of $G$ clipped to $S$ is at most $cr$.
\end{definition}

Unless otherwise specified,
we call the number of vertices of the graph $n$ and its number of edges $m$.
Note that both $\lambda$ and~$c$ are upper bounds for the maximum degree of a
graph with $\lambda$-low density and packedness~$c$,
so $m\in\O(n)$ if $\lambda$ or $c$ is constant.
It was already known that each graph has $\lambda\le2c$
\cite{driemel2010approximating},
but for the other direction we give the following new result.
\begin{lemma}
\label{lem:relationship}
Let $G$ be a graph with $\lambda$-low density,
then $G$ has packedness $c\in\O\bigl(\lambda\sqrt n\bigr)$.
\end{lemma}
\begin{proof}
Let $S$ be an arbitrary square of radius $r$.
W.l.o.g.\ assume $r=1$; else scale $G$ and~$S$ by~$1/r$.
We want to show that the total length of edges inside $S$
is $\O\bigl(\lambda\sqrt n\bigr)$.
There are at most $\lambda$ edges
whose intersection with $S$ has length greater than $2$;
this intersection has length at most $2\sqrt2$.
Let $k$ be some natural number to be determined shortly.
For each $i\in\{0,\ldots,k\}$ and
for each of the $4^i$ squares of radius $2^{-i}$ that cover $S$,
there can be at most $\lambda$ edges
whose intersection with $S$ has length between $2^{-i}$ and $2^{-i+1}$.
Generally there can be at most $\lambda n$ edges in $G$
since it has maximum degree $\lambda$.
In particular there can be at most $\lambda n$ edges
of length at most~$2^{-k}$ intersecting~$S$.
This covers all intersections of edges with~$S$.
Their total length is
\[\|G\cap S\|
\le\lambda2\sqrt2+\sum_{i=0}^k4^i\lambda2^{-i+1}+\lambda n2^{-k}
\le4\lambda\bigl(1+2^k+n2^{-k}\bigr).\]
With $k:=\bigl\lfloor\log\sqrt n\bigr\rfloor$ the claim follows.
\end{proof}
This bound is tight:
A $\bigl(\sqrt n\times\sqrt n\bigr)$-grid graph
has constant density but packedness~$\Theta(\sqrt n)$.

\section{Balanced Separators}%
\label{sec:bal-sep}
Balanced separator computation is a core primitive in many algorithms. 
An $\alpha$-balanced separator in $G(V,E)$ for $\alpha \in (0,1)$ is a set of
nodes $S \subset V$ such that all connected components in $G[V\setminus S]$
contain at most $\alpha n$ nodes.
Computing the smallest $\alpha$-balanced separator for given $\alpha$ takes
exponential time in general graphs, but in $c$-packed and $\lambda$-low-density graphs efficient algorithms with parameterized separator size bounds are known. However,  the hidden constants in the respective separator sizes are too large to
be useful in practice.
We thus device modified   separator constructions with
significantly improved separator sizes.

\subsection{From Ball Separators to Square Separators in Geometric Graphs}
\label{sec:balltosquare}
In \cite{le2024greedy} a randomized algorithm for computing balanced separators
in $\lambda$-low-density graphs\footnote{The paper actually considers
$\tau$-lanky graphs and proves a separator size in $\O(\tau \sqrt n)$.
However by definition, $\tau \leq \lambda$ and thus the result directly
transfers to low-density graphs.} was proposed which guarantees a separator size
in $\O(\lambda \sqrt n)$.
They prove that there always exists a ball $B_r$ of radius $r$ that encloses a
constant fraction of all nodes but not more~than~$\frac n2$, such that the
number of edges of length at most $r$ intersecting the ball is in
$\O(\lambda \sqrt n)$.
As there are at most $\lambda$ edges of length greater~than~$r$
that intersect the ball by
definition, the respective separator size follows.
They argue that a ball with these properties can be computed in the Euclidean
plane with an expected running time in $\O(\lambda \sqrt n)$.
For $c$-packed graphs a deterministic algorithm was proposed in
\cite{deryckere2025wspd} that computes a balanced separator of size
$\O(c)$ in time $\O(c^2 n)$.
The idea is to compute two nested balls $B=B_r$ and $B'=B_{2r}$, with $B$
enclosing at least a constant fraction of the nodes, and at least half of the
nodes being outside of $B'$.
Then running a max-flow algorithm from border nodes inside $B$ to nodes outside
$B'$ and using the max flow--min cut duality, a valid balanced separator is
identified. See \cref{fig:orig-sep} for an illustration of this algorithm.
\begin{figure}[ht]
        \begin{subfigure}[t]{0.32\textwidth}
        \centering
        \includegraphics[page=2,width=\textwidth]{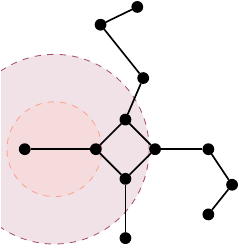}
        \caption{The first step consists in finding two nested balls ${\color{kn_peach}B}$ and ${\color{kn_bordeaux65}B'}$ enclosing a particular amount of vertices.}
        \label{fig:orig-sep-step-balls}
    \end{subfigure}\hfill
    \begin{subfigure}[t]{0.32\textwidth}
        \centering
        \includegraphics[page=5,width=\textwidth]{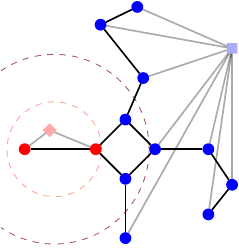}
        \caption{The second step is to solve a max-flow instance where sources are nodes inside ${\color{kn_peach}B}$, here connected to the supersource ${\color{red}\blackdiamond}$, and sinks are nodes outside ${\color{kn_bordeaux65}B'}$, again connected to the supersink ${\color{blue}\blacksquare}$. Edge capacity is set to 1.
        From this, the min-cut $({\color{red}S},{\color{blue}T})$ is extracted.}
        \label{fig:orig-sep-step-flow}
    \end{subfigure}\hfill
    \begin{subfigure}[t]{0.32\textwidth}
        \centering
        \includegraphics[page=7,width=\textwidth]{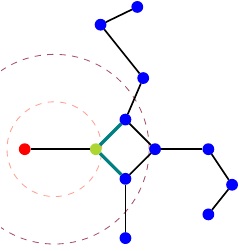}
        \caption{Finally, {\color{teal}cut edges} can be identified, and the {\color{lime!60!black!90}separator} is defined as a vertex cover of the latter.}
        \label{fig:orig-sep-step-sep}
    \end{subfigure}
    \caption{Illustration of the original separation method from \cite{deryckere2025wspd}.}
    \label{fig:orig-sep}
\end{figure}
They argue that for each separator node there needs to be a path of length at
least $2r-r=r$ from inside $B$ to outside $B'$ and the set of such paths needs
to be edge disjoint.
As the total length of edges inside $B'$ is bounded by $2cr$ by definition of
packedness, an upper bound of~$2c$ on the separator size follows.
Given our newly established relationship between $c$ and $\lambda$ in
\cref{lem:relationship}, we get the following corollary.
\begin{corollary}
    Each separator of size $\O(c)$ has size
    $\O\bigl(\lambda\sqrt n\bigr)$.
\end{corollary}
This directly implies that the algorithm for $c$-packed graphs works for
$\lambda$-low-density graphs with a  size bound matching the one shown
in \cite{le2024greedy}, but without relying on randomization.

\begin{table}[t]
    \centering
    \begin{tabular}{|l|rc|r|}
    \hline
    method & $\alpha$ & $|S_\alpha|$ & $|S_{2/3}|$\\
    \hline
    nested balls & $1-\frac1{2 \delta^3} \approx 0.998$ & $2c$ & $406c$ \\
    + better balance (bb) & $1-\frac1{\delta^3+1} \approx 0.997$ & $2c$ & $270c$ \\
    + approximation (apx) & $1-\frac1{\delta+1} = 0.875$ & $\frac2{1-\varepsilon}c$ & $9c$\\
    nested squares + bb + apx & $1-\frac15 = 0.800$ & $\frac{2\sqrt2}{1-\varepsilon}c$ & $6c$\\
    \hline
    \end{tabular}
    \caption{Comparison of $\alpha$-balanced separator computation methods for
    $c$-packed graphs.
    The specified $\alpha$ is the provided balance guarantee and $|S_\alpha|$
    denotes the respective separator size.
    The column $|S_{2/3}|$ shows the separator size resulting from iterating the
    method until each component contains at most $\frac23n$ nodes
    (using $\varepsilon = 0.05$ for the bottom two rows).}
    \label{tab:sepsize}
\end{table}
In principle the separator computation algorithm for $c$-packed graphs as
presented in~\cite{deryckere2025wspd} is practical.
However there is one caveat:
The balance factor that is guaranteed by the approach is
$\left(1-\frac1{2\delta^3}\right)$ where $\delta$ denotes the doubling
constant of $\mathbb{R}^d$.
For $d=2$ we have $\delta = 7$ and therefore the achieved balance factor is
greater~than~$0.998$.
To obtain a smaller balance factor $\alpha$, this algorithm can be applied
iteratively to the largest connected component of the graph after separator
removal.
If we use $k$ iterations, the resulting balance factor guarantee is $\alpha
= \left(1-\frac1{2\delta^3}\right)^k$ with a running time in
$\O(c^2nk)$.
However the resulting separator consists in the worst-case of the union of all
$k$ computed separators, increasing the separator size from $2c$ to $2kc$.
Thus to achieve $\alpha =\frac12$, we have to set $k \approx 475$ and
therefore get a separator size bound of approximately $950c$.
With such a large coefficient, the practical utility is void.
In the following we will discuss
how to decrease the coefficient significantly by
(i)~better balancing the number of nodes inside $B$ and outside $B'$,
(ii)~leveraging $(1+\varepsilon)$-approximation algorithms for smallest
enclosing ball,
(iii)~switching from balls to squares.
\Cref{tab:sepsize} summarizes the (accumulated) impact of these
improvements for the construction of a $\frac23$-balanced separator.

\bigskip\noindent\textit{Better Balance.}
The balance guarantee of $\alpha = \left(1-\frac1{2\delta^3}\right)$ shown in
\cite{deryckere2025wspd} relies on the following approach: $B$ is supposed to
enclose $k=\frac n{2\delta^3}$ nodes.
However computing the smallest such ball with radius~$r_{opt}$ takes time in
$\O(nk)$ \cite{harpeled2005fast} and thus for the chosen $k$ this results in
a quadratic running time.
Therefore a linear time algorithm is applied instead, which returns a ball~$B$
with radius~$r \leq 2 r_{opt}$ \cite{harpeled2005fast}.
Accordingly each ball of radius $\frac r2$ contains strictly less than $k$
nodes.
Then $B'$ is constructed from $B$ by doubling its radius, that is, $r'=2r \leq
4 r_{opt}$.
Consequently $B'$ can be covered with at most $\delta^3$ balls of radius
$\frac r2$.
As a result $B'$ contains less than $k\delta^3 = \frac n2$ nodes.
But the worst balance factor is anyway obtained in case only the $k$ nodes in
$B$ are separated from the rest.
Thus it suffices to ensure that there are at least as many nodes outside of
$B'$ as inside of $B$, that is, $n-\delta^3k = k$ which yields
$k=\frac n{\delta^3+1}$. The respective improved balance factor shown in
the second row of \cref{tab:sepsize}.

\bigskip\noindent\textit{Smallest Enclosing Ball Approximation.}
To further improve $\alpha$, we leverage the fact that in~\cite{harpeled2005fast} not
only a $2$-approximation for smallest enclosing ball of $k$ points
is presented but also a $(1+\varepsilon)$-approximation which runs in time
$\O(n + \frac1{\varepsilon^3} \log^2 \frac1\varepsilon)$ for
$\varepsilon > 0$ and $k \in \Omega(n)$.
We use this algorithm to compute a ball $B$ with radius $r \leq (1+\varepsilon')
r_{opt}$ for some $\varepsilon' > 0$.
We then choose $\varepsilon \in (\varepsilon',1)$.
Based thereupon, we know that $r/(1+\varepsilon) < r_{opt}$.
We set the radius of $B'$ to $r'=2 r/(1+\varepsilon)$.
By the ball covering property, it directly follows that $B'$ contains at most
$\delta k$ nodes.
To balance the number of nodes inside $B$ and outside $B'$, we solve $n-\delta k
= k$ for~$k$ and get $\frac n{\delta+1}$, which results in a significant
improvement in balance factor; see the third row of \cref{tab:sepsize}.
However these modifications also affect the separator size.
We know that each path from inside $B$ to outside of $B'$ has a length of at
least $r'-r$ and the total length of all edge disjoint paths is upper bounded by
$cr'$.
With $r'=2 r/(1+\varepsilon)$ we get $\frac{cr'}{r'-r}=
\frac{2c}{1-\varepsilon}$.

\bigskip\noindent
\textit{Switching to Squares.}
To cover a ball in the plane, $\delta=7$ balls of half the radius are needed.
However for a square with side length $r$, we only need $4$ squares of side
length $\frac r2$ to cover it.
Thus if we can compute a $(1+\varepsilon)$-approximation for the smallest
enclosing square of $k$ points, the analysis from above can be transferred with
two differences:
If the definition of $c$ relies on balls, then the segment lengths inside a
square of side length $r$ is bounded by $\sqrt2cr$ instead of $cr$ and thus the
separator size bound increases proportionally.
The second difference is that we can now set $k=\frac n5$ as $B'$ will
contain at most $4k$ nodes and thus the number of nodes outside of $B'$ will
also be at least $\frac n5$.
This further improves the balance factor; see the fourth row in
\cref{tab:sepsize}.
In \cite{chan2021smallest}, the smallest enclosing rectangle problem was
considered and a $(1+\varepsilon)$-approximation for the rectangle area was
described which runs in $\O(n \log n \frac1{\varepsilon^3} \log
\frac1\varepsilon)$.
By restricting the algorithm to consider squares, the area approximation
translates into an approximation of side length, which is exactly what we need.
However, by combining ideas from \cite{harpeled2005fast} and~\cite{chan2021smallest},
we can design an algorithm directly tailored to side length approximation for
smallest enclosing square that runs in $\O(n + \frac1{\varepsilon^2}
\log \frac1\varepsilon)$ for $\varepsilon > 0$ and $k \in \Omega(n)$ and is
thus faster than the algorithms for balls and rectangles.
The detailed algorithm description is provided in \cref{apx:enclosing-square}.
Given that for any fixed $\varepsilon > 0$ the computation of the required
square is in~$\O(n)$, the running time of the separator computation is
still dominated by the max-flow component.
Therefore we match the $\O(c^2 n)$ running time of the nested ball
approach.
But now, we can construct an $\alpha=\frac23$-balanced separator with a
practically useful size guarantee of $6c$ when using $\varepsilon = 0.05$, which
is over a factor of 60 smaller than the original bound.
This factor further increases for smaller choices of $\alpha$.

\section{Separator-Based Exact Distance Oracles}%
\label{sec:edo}
Given an algorithm that computes balanced separators, an exact distance oracle
(EDO) can be constructed by computing a balanced hierarchical
separator tree $\mathcal{S}$ and storing with each node per separator its
shortest path distances to all nodes in the respective subgraph.
See \cref{fig:sep-tree} for an illustration of $\mathcal{S}$.
We refer to this data structure as $\mathcal{S}$-EDO.
In an $a$-$b$-query, first the separator bags in $\mathcal{S}$ containing $a$
and $b$ are identified, respectively.
Then for nodes $w$ in the union of all bags
from their lowest-common-ancestor bag
to the root of $\mathcal{S}$,
distances $d(w,a)$ and $d(w,b)$ are looked up and the smallest
sum of such distances is tracked.
As balanced separators are used for the construction of $\mathcal{S}$, its depth
is logarithmic.
Thus if every separator in $\mathcal{S}$ has a size in~$\O(c)$, the
space consumption is in $\O(cn \log n)$ and the query time is in
$\O(c \log n)$.
The data structure can be constructed in $\O(c^2 n \log^2 n)$ time \cite{deryckere2025wspd}.

When given a tree decomposition $\mathcal{T}$ with bag size \tw, though, an
alternative distance oracle ($\mathcal{T}$-EDO) was described in
\cite{conrado2024faster}.
The construction time is in $\O(\tw^3 n \log n)$, the space consumption
in $\O(\tw \cdot n \log n)$, and the query time in $\O(\tw)$.
It is known that $c$-packed graphs have treewidth $\tw \in \O(c)$
\cite{dvorak2019treewidth,deryckere2025wspd}.
However constructing a tree decomposition with maximum bag size \tw\ takes
$2^{\O(\tw^2)}n^{\O(1)}$ time \cite{korhonen2023improved},
and no constant factor approximation algorithm for general graphs exists
unless $\mathsf P=\mathsf{NP}$~\cite{wu2014inapproximability}.
In this section we show how to compute $\mathcal{T}$ for $c$-packed graphs of
width $\O(c)$ in polynomial time.
This makes a rich set of algorithms based on tree decompositions directly
applicable to $c$-packed graphs, including the improved distance oracle
described in \cite{conrado2024faster}.

\subsection{Tree Decompositions from Separators}
A simple way to obtain a valid tree decomposition $\mathcal{T}$ using a balanced
separator oracle is as follows: First construct a balanced hierarchical
separator tree $\mathcal{S}$, subsequently include in each separator bag $X$ all
bags on the path from $X$ to the root of $\mathcal{S}$.
This ensures that each edge is covered by at least one bag and that the bags
containing a specific node form a connected subtree.
However, if $B$ is an upper bound on the separator size occurring in
$\mathcal{S}$, the resulting width of the tree decomposition is in
$\O(B\log n)$.

Reed's algorithm \cite{reed1992finding} for tree decomposition construction in
graphs with bounded separator size avoids this logarithmic blow-up
The algorithm takes as input a graph $G(V,E)$ with an assumed upper bound $B$ on
its treewidth \tw, as well as a special node subset $W$ of size at most $3B$.
It then either certifies that $\tw>B$ or returns a tree decomposition with
width at most $4B$.
When calling~$\reed(G(V,E),W)$, the algorithm first computes a so-called balanced
$W$-separator~$S$ which must ensure for each connected component $C$
in~$G[V\setminus S]$ we have $|C| \leq \frac34|V|$ and $|C \cap W| \leq
\frac23|W|$.
Thus $S$ must separate both the whole graph and the subset~$W$ in a balanced
manner at the same time.
Reed shows that such a balanced $W$-separator of size at most $B$ can be
computed in time $\O(9^B B |E|)$ if it exists.
Once a separator of suitable size is identified, the algorithm creates a tree
decomposition bag $S \cup W$ and recursively calls itself on each connected
component $C$ of $G[V\setminus S]$ using~$\reed(G[C \cup S], (C\cap W) \cup S)$.
See \cref{fig:reed} for an illustration of this process.
\newsavebox{\leftbox}%
\newsavebox{\rightbox}%
\savebox{\leftbox}{\includegraphics[page=17,width=0.54\textwidth]{graphics/schematics.pdf}}%
\savebox{\rightbox}{\includegraphics[page=18,width=0.44\textwidth]{graphics/schematics.pdf}}%
\begin{figure}[ht]
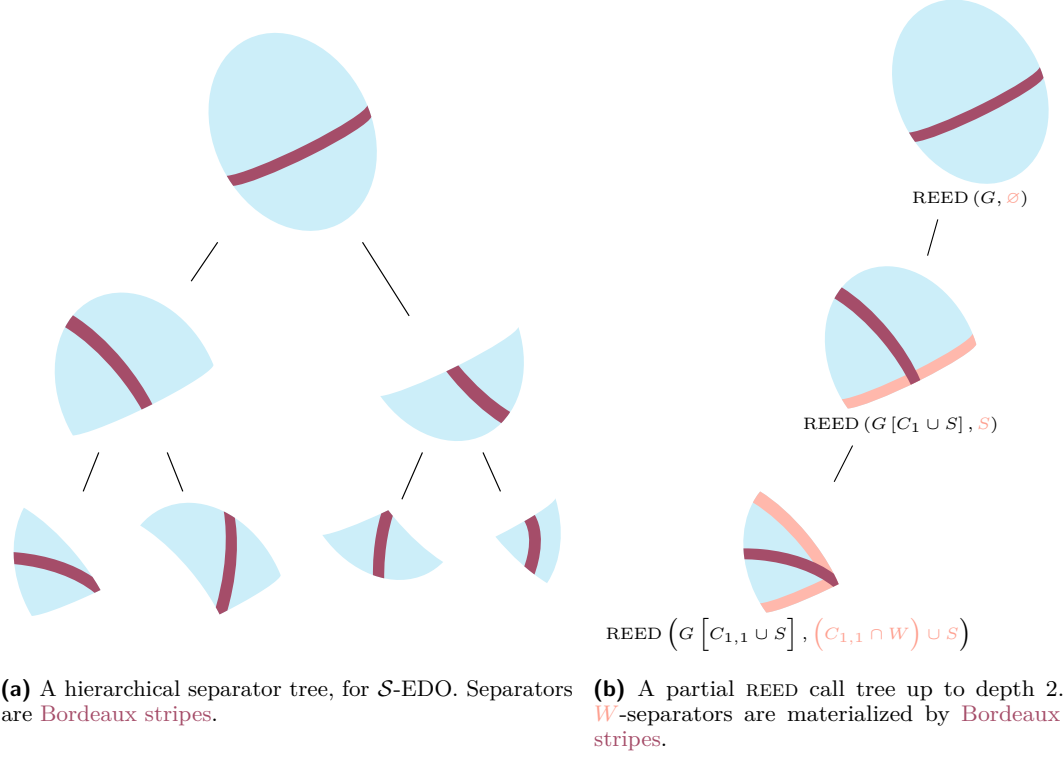

    \begin{subfigure}[t]{0.54\textwidth}
        \centering
        \raisebox{0.5\dimexpr\ht\rightbox - \ht\leftbox\relax}{\usebox{\leftbox}}
        \caption{A hierarchical separator tree, for $\mathcal{S}$-EDO. Separators are {\color{kn_bordeaux65}Bordeaux stripes}.}
        \label{fig:sep-tree}
    \end{subfigure}\hfill
    \begin{subfigure}[t]{0.44\textwidth}
        \centering
        \usebox{\rightbox}
        \caption{A partial $\reed$ call tree up to depth 2. ${\color{kn_peach}W}$-separators are materialized by {\color{kn_bordeaux65}Bordeaux stripes}.}
        \label{fig:reed}
    \end{subfigure}
    \caption{Illustration of the two main constructs leading to Exact Distance Oracles.}
    \label{fig:s-t-structs}
\end{figure}
If the number of nodes in $G[C \cup S]$ drops below $4B$, no further refinement
is necessary and instead a tree decomposition bag containing $C \cup S$ is
returned.
Tree decompositions resulting from recursive calls on the components $C$ are
then connected to the parent bag $S \cup W$.
As we always have $|S| \leq B$ and $W \leq 3B$, the maximum bag size is $4B$.
The size bound of $W$ is maintained throughout the recursive calls as for each
component $C$ we have $|C \cap W| \leq \frac23|W|$ by definition of the
balanced $W$-separator and thus $|(C \cap W) \cup S| \leq \frac23 3B + B
\leq 3B$.
The size of the graph in a recursive call is bounded by~$|C \cup S| \leq
\frac34|V| +B $.
Therefore the graph size drops below $4B$ after a logarithmic number of calls.
Thus Reed's algorithm constructs a tree decomposition with a width linear in
$B$.
However the balanced $W$-separator computation---with exponential dependency
on~$B$---only runs in polynomial time on graphs with $\tw\le B\in\O(1)$.

\subsection{Application to \texorpdfstring{$c$}{c}-Packed Graphs}%
\label{subsec:tree_dec_cpacked}
The naive tree decomposition construction with width in $\O(B\log n)$
is directly applicable to $c$-packed graphs using $B \in \O(c)$.
The construction time for $\mathcal{S}$ is in $\O(c^2 n \log n)$, and
the $\mathcal{T}$~construction based thereupon runs in time $\O(cn \log
n)$.
However, with a width of~$\O(c \log n)$, the query time of
$\mathcal{T}$-EDO is the same as for $\mathcal{S}$-EDO.

To get $\mathcal{T}$ with width $\O(c)$, we now adapt Reed's algorithm
to $c$-packed graphs.
The crucial step is the efficient computation of a balanced $W$-separator $S$
for a given graph~$G(V,E)$ and a node subset~$W$, such that each component in
$G[V\setminus S]$ contains at most a constant fraction of the nodes in $V$ and
the nodes in $W$.
To achieve this, we actually compute two separators $S_V$ and $S_W$ and set $S:=
S_V \cup S_W$.
Here $S_V$ is simply a $\frac45$-balanced square separator of~$G$ of size
at most $3c$ constructed as described in \cref{sec:balltosquare} with
$\varepsilon=0.05$.
For~$S_W$ we require that it separate both $G$ and $W$ simultaneously, but
only for $W$ the balance criterion needs to be enforced.
This can be achieved by constructing the two nested squares for separator
computation solely based on the nodes in $W$, but then considering all nodes in
$V$ for the actual flow and separator computation.
See \cref{fig:w-sep} for an illustration.
\begin{figure}[ht]
    \centering
    \includegraphics[page=19,width=.25\linewidth]{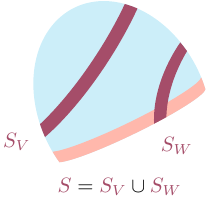}
    \caption{Illustration of the balanced $W$-separation process using $c$-Packed separation.}
    \label{fig:w-sep}
\end{figure}
As the separator bound of $3c$ is oblivious to the number of nodes enclosed in
the respective squares but only depends on their side length ratio, we get the
same maximum separator size for $S_W$ as for $S_V$.
Accordingly we have $|S| \leq |S_V| + |S_W| \leq 6c$.
It remains to adapt the bound on the size of $W$ to make the recursion work.
If $|W| \leq 30c$, we know that $|(C \cap W) \cup S| \leq \frac45 |W| + 6c
\leq 30c$.
Furthermore $|C \cup S| \leq \frac45 |V| + 6c$ holds, and therefore, after a
logarithmic number of steps, a subgraph size of at most $30c$ is reached.
We remark that we do not require explicit knowledge of the value of $c$ to stop
the recursion.
Instead we can keep track of the largest bag size we have seen so far and once
the current subgraph size drops below this threshold, we can safely abort.
As the recursion depth is logarithmic and the separator computations at each
recursion level take $\O(c^2 n)$ time in total, we get the following
theorem.
\begin{theorem}
    Given a $c$-packed graph with $n$ nodes, a tree decomposition of width $30c$
    can be computed in time $\O(c^2 n \log n)$.
\end{theorem}

The theorem allows to apply many known algorithms based on tree decompositions
to $c$-packed graphs after a polynomial time preprocessing.
This includes the exact distance oracle $\mathcal{T}$-EDO as described in
\cite{conrado2024faster}.
\begin{corollary}
    Given a $c$-packed graph, an exact distance oracle can be computed in
    $\O(c^3 n \log n)$ time, with a space consumption in $\O(c
    n \log n)$, and a query time in $\O(c)$.
\end{corollary}
Thus, compared to the $\mathcal{S}$-EDO data structure described in
\cite{deryckere2025wspd}, the $\mathcal{T}$-EDO data structure for $c$-packed
graphs is factor of $c$ slower in preprocessing time but a factor of $\log n$
faster with respect to query time, achieving query times that are independent of
the overall graph size.

\section{Experimental Results}
\label{sec:expe_results}
In this section, we first devise engineered algorithms for the computation of $\lambda$ and $c$.
Afterwards, we describe our implementations of the balanced separator and the tree decomposition algorithm in detail and evaluate them on  real-world road networks extracted from OSM, see \cref{tbl:instances} for selected instances.
Experiments are conducted on an AMD Ryzen Threadripper 1950X 16-Core Processor clocked at \SI{3.4}{\giga\hertz} with \SI{128}{\giga\byte} of~RAM.

\begin{table}[htb]
    \centering
    \begin{tabular}{lrrrrr}
         \bf Instance&
         $n=\bf\# vertices$&
         $m=\bf\#edges$&
         \bf density $\lambda$&
         \bf packedness $c$&
         $\lambda\sqrt n$\\
         \hline
 OSM-01 & \num{500}      & \num{540}      & \num{9}  & $[32,64]$     & \num{201}    \\
OSM-02 & \num{1000}     & \num{1084}     & \num{9}  & $[37,74]$     & \num{285}    \\
OSM-03 & \num{2000}     & \num{2190}     & \num{14} & $[45,90]$     & \num{626}    \\
OSM-04 & \num{5000}     & \num{5407}     & \num{14} & $[61,122]$    & \num{990}    \\
OSM-05 & \num{10000}    & \num{10748}    & \num{14} & $[78,156]$    & \num{1400}   \\
OSM-06 & \num{20000}    & \num{21394}    & \num{14} & $[103,206]$   & \num{1980}   \\
OSM-07 & \num{50000}    & \num{53700}    & \num{15} & $[162,324]$   & \num{3354}   \\
OSM-08 & \num{100000}   & \num{107180}   & \num{17} & $[224,448]$   & \num{5376}   \\
OSM-09 & \num{200000}   & \num{214178}   & \num{17} & $[316,632]$   & \num{7603}   \\
OSM-10 & \num{500000}   & \num{535115}   & \num{29} & $[475,950]$   & \num{20506}  \\
OSM-11 & \num{1000000}  & \num{1062928}  & \num{29} & $[597,1194]$  & \num{29000}  \\
City   & \num{121936}   & \num{231589}   & \num{46} & $[444,888]$   & \num{16063}  \\
State  & \num{346127}   & \num{680991}   & \num{34} & $[639,1278]$  & \num{20003}  \\
Cutout & \num{999591}   & \num{2120808}  & \num{64} & $[1447,2894]$ & \num{63987}  \\
Small  & \num{100242}   & \num{212220}   & \num{30} & $[452,904]$   & \num{9498}   \\
Medium & \num{4094608}  & \num{8287301}  & \num{96} & $[2120,4240]$ & \num{194257} \\
Large  & \num{20690320} & \num{21895407} & \num{48} & $\ge2643$ & \num{218336} \\
    \end{tabular}
    \caption{Characteristics of our benchmark graphs.}
    \label{tbl:instances}
\end{table}

\subsection{Efficient Computation of the Density}
Both the exact algorithm~\cite{deberg2002realistic}\footnote
{The paper states it takes as input a ``planar scene.''
\cite{gudmundsson2023approximatinglambda} interpreted this
to mean that no edge crossings are allowed.
However it only means that the objects lie in the plane (as in $\R^2$),
since the algorithm and its proof do not make
any mention of planarity (as in freedom of intersections).}
and the 3-approximation~\cite{gudmundsson2023approximatinglambda} for computing $\lambda$
use the same framework:
For each edge~$e$ we find
the set~$R(e)$ of edges that lie within~$2\|e\|$ of~$e$;
this takes time $\O(m\log m)$ with a quadtree.
For each edge $e$ and
each $p$ from a certain set~$P(e)$ of points,
we count the edges from $R(e)$ that lie within~$\|e\|$ of a point~$p$;
this is equivalent to a stabbing query
with input $\bigl\{\,e'\oplus S(\|e\|)\mid e'\in R(e)\,\bigr\}$,
where $\oplus$ denotes the Minkowski sum
and $S(r)$ the square (or ball) of radius~r about the origin.
Finally we return the maximum of these counts over all choices of $e$ and~$p$.
For each $e\in E(G)$,
the 3-approximation chooses $|P(e)|=48\,600\in\O(1)$ points,
while the exact algorithm computes
the arrangement of~$|R(e)|\in\O(\lambda)$ Minkowski sums
in time $\O\bigl(\lambda^2\bigr)$---at least on paper.
In our implementation we simply use
the $\O\bigl(\lambda^2\bigr)$ intersection points of the boundaries of
these Minkowski sums as~$P(e)$,
which results in a running time of $\O\bigl(\lambda^3\bigr)$ per choice of $e$
or $\O\bigl(\lambda^3m\bigr)$ in total.
However we believe
that the hidden constants in the running time of the arrangement construction
are larger than $\lambda$ for realistic instances.

We can organize these computations in several ways.
First we can process the edges $e$ by decreasing length
as suggested by~\cite{deberg2002realistic};
then we know that the current~$e$ admits a new count of at~most~$\lambda'+1$,
where $\lambda'$ is the previous maximum;
thus we can skip the remainder of $P(e)$ as soon as we find a point that lies
close to $\lambda'+1$ edges from $R(e)$.
Additionally we can skip~$e$ entirely if $|R(e)|\le\lambda'$;
in this case we cannot hope to find a new maximum;
this is ``full synchronization.''
We can hope to exploit the latter shortcut more often
if we process the edges~$e$ by decreasing $|R(e)|$.
Indeed, if we abandon the former shortcut entirely,
we need not wait for the result of the stabbing query for one choice of~$e$
before we start with the next;
that is, we may run the stabbing queries concurrently;
optionally with ``limited synchronization'' to skip edges~$e$ with small~$R(e)$.
This way the computation can run on up to $m$ CPU cores.
\begin{figure}
    \includegraphics[width=.5\textwidth,page=1]{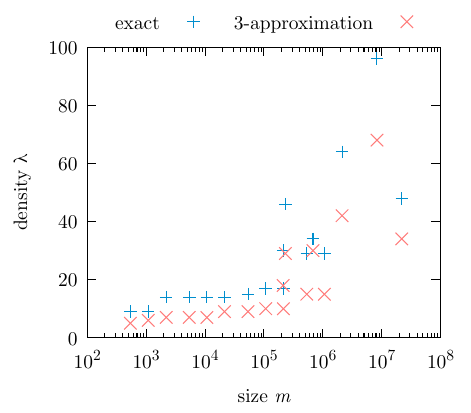}%
    \includegraphics[width=.5\textwidth,page=2]{density}
    \caption{The low-density constants of several road networks as computed by
    the exact algorithm and the 3-approximation as well as the respective
    running times.}
    \label{fig:density-results}
\end{figure}

When one examines the results in \cref{fig:density-results},
it is immediately clear that the 3-approximation
is actually slower than the exact algorithm.
We attribute the difference in running time to the fact
that all instances have relatively low density,
so $\lambda^3$ is in fact smaller than $48\,600\,\lambda$.
Additionally we observe that the 3-approximation is
usually worse than the optimum by a factor of about~1.5.
Finally we see that in all cases the variant
that runs on multiple cores with limited synchronization is the fastest.

\subsection{Efficient Computation of the Packedness}
For the packedness $c$, we implement the 2-approximation from~\cite{gudmundsson2023approximatingc} and the randomized  42-approximation approximation from \cite{harpeled2025how}.
The 2-approximation restricts the centers of the squares it considers to vertices of the graph
and sweeps outward from each vertex. It runs in~$\O(m^2\log m)$.
A key detail of this algorithm is the update operation it performs whenever the
sweep line hits an event point:
It updates the growth rate of the total length $\ell$ of edges inside the square
per additional unit of radius (i.e., $d\ell/dr$) in constant time.
An implementer must take care
that floating point imprecisions do not accumulate too much.
To guarantee this, we use Kahan's summation algorithm~\cite{kahan1965pracniques}
in our implementation.
Furthermore, since we run a sweep from each vertex independently,
the algorithm is embarrassingly parallel;
only  combining maximum accumulation at the end requires synchronization.

The randomized 42-approximation runs in $\O(m\log^3m)$~\cite{harpeled2025how}.
This algorithm uses a WSPD to find a linear number of squares
one of which approximates the packedness.
Then it randomly scales and shifts these squares $\Theta(n)$ times,
builds a quadtree of the translated squares,
and approximates the packedness of its nodes in a bottom-up traversal.
With high probability a square with high packedness
is approximately equal to a node in one of the quadtrees.
We may once again create and process the quadtrees
in parallel on up to $\Theta(\log n)$ cores.
To make the 42-approximation algorithm for packedness scalable to large graphs, we make the following adjustments:
 In practice the separation factor 7 200 of the WSPD as specified in the
paper is prohibitively large for instances of practical sizes. Other constants that control the approximation factor for the packedness of the quadtree lead to impractical running times as well. We changed them to the values shown in Table \ref{tbl:packedness-constants}. The result is a worse approximation
guarantee in theory, but our experiments show that the approximation factor is still always
below 5 in practice. 
\begin{table}[b]
\center
\begin{tabular}{lll}
	\bf Constant&
	\bf Value from paper&
	\bf Our implementation
\\ \hline
	$\alpha$&
	20 \cite[lemma~22]{harpeled2025how}&
	3
\\
	$\xi$&
	1/300 \cite[pages 10--11]{harpeled2025how}&
	1/30
\\
	$s$&
	7200 \cite[page~7]{gudmundsson2023approximatingc}&
	6
\\
	$m=\#{\rm quadtrees}$&
	$\O(1000\log n)$ \cite[lemma~10]{harpeled2025how}&
	$\lfloor\log_2 n\rfloor$
\end{tabular}
\caption{Constants from the 42-approximation as specified in the paper
	versus the values we use in our implementation.}
\label{tbl:packedness-constants}
\end{table}
\begin{figure}
    \includegraphics[width=.5\textwidth,page=1]{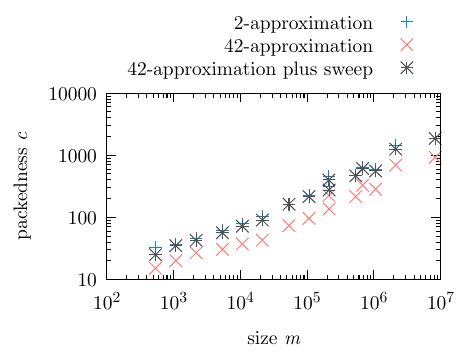}%
    \includegraphics[width=.5\textwidth,page=2]{congestion}
    \caption{The packednesses of several road networks as computed by
    the 2- and 42-approximation as well as the respective
    running times.}
    \label{fig:congestion-results}
\end{figure}

In \cref{fig:congestion-results} we can see
that the 42-approximation is much faster than the 2-approximation,
for large instances even by several orders of magnitude.
This comes at the cost of solution quality:
The 2-approximation usually finds a square that is twice as packed
as the estimate from the 42-approximation.
We can further increase the quality of the result in practice by computing the
exact packedness of the square output by the 42-approximation by running a
single sweep from its center as in the 2-approximation.
This takes only $\O(m\log m)$ time,
so it does not change the asymptotic running time;
but with this additional step we find squares
whose packedness is at~least \SI{75}{\percent} of those found by the 2-approximation
for our test instances,
and indeed the running time of the sweep is negligible
in comparison to the 42-approximation.
This way we get the best tradeoff between quality and running time.

For the two largest instances,
where the 2-approximation did not finish within twelve hours,
we report lower bounds on $c$
from the 42-approximation and partial runs of the 2-approximation
in \cref{tbl:instances}.
As one can see there, $c$ is always much larger than $\lambda$ for
the same graph, but also much smaller than  $\lambda \sqrt n$, indicating that
the known theoretical relationships between the two parameters are far from
tight on the tested road networks.

\subsection{A Practical Separation Algorithm}

One can conceptualize the separation algorithm proposed in \cite{deryckere2025wspd} as two independent steps: first, run a strategy yielding two balls which ensure a given balance factor, and second, extract vertices inside $B$ and outside~$B'$, and then run  a max-flow min-cut algorithm. Extracting vertices is done using a naive linear filtering. The chosen max-flow algorithm is Ford--Fulkerson with iterative depth-first searches to find augmenting paths, resulting in a running time of $\O(F \cdot (n+m))$ where $F$ is the maximum flow, which reduces to $\O(c^2 n)$ for $c$-packed graphs.
Because the first step is more flexible, we could use any of the methods shown in \cref{tab:sepsize}.
As a baseline and suggested in \cite{deryckere2025wspd}, and in order to compute $B = B_r$ and deduce $B' = B_{2r}$, the $\O(n(\frac{n}{k})^d)$-time 2-approximation for $k$-enclosing disk from \cite{harpeled2005fast} was implemented. Using $k = \frac{n}{2\lambda^3}$, it indeed yields a linear time algorithm, but in practice constant factors are substantial: in dimension~2, the running time is $\num{470596} n$.
Besides, as presented in \cref{sec:bal-sep}, the best balance factor using this 2-approximation is $0.997$ with \emph{Better Balance}. In practice, it yields exactly this balance; see \cref{fig:illus_sep_small_original_bb}.
\begin{figure}[ht]
 \begin{subfigure}[t]{0.48\textwidth}
        \centering
        \begin{tikzpicture}
            \node[inner sep=0, anchor=south west] (background) {\includegraphics[trim={5cm 0cm 5cm 0cm},clip,width=\textwidth]{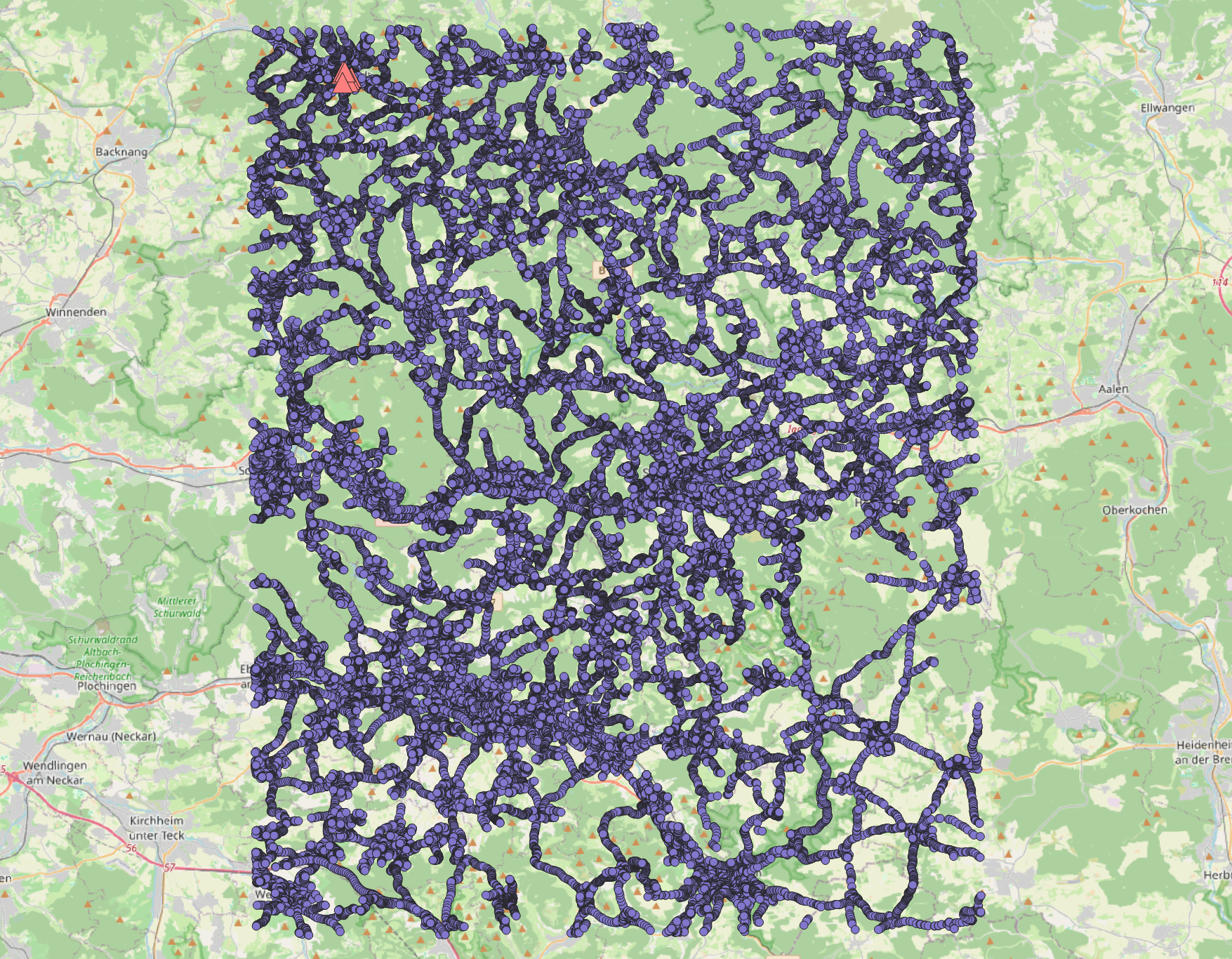}};
            \node[inner sep=0, anchor=south east] at (background.south east) {\includegraphics[trim={4cm 0cm 3cm 1cm},clip,width=.75\textwidth]{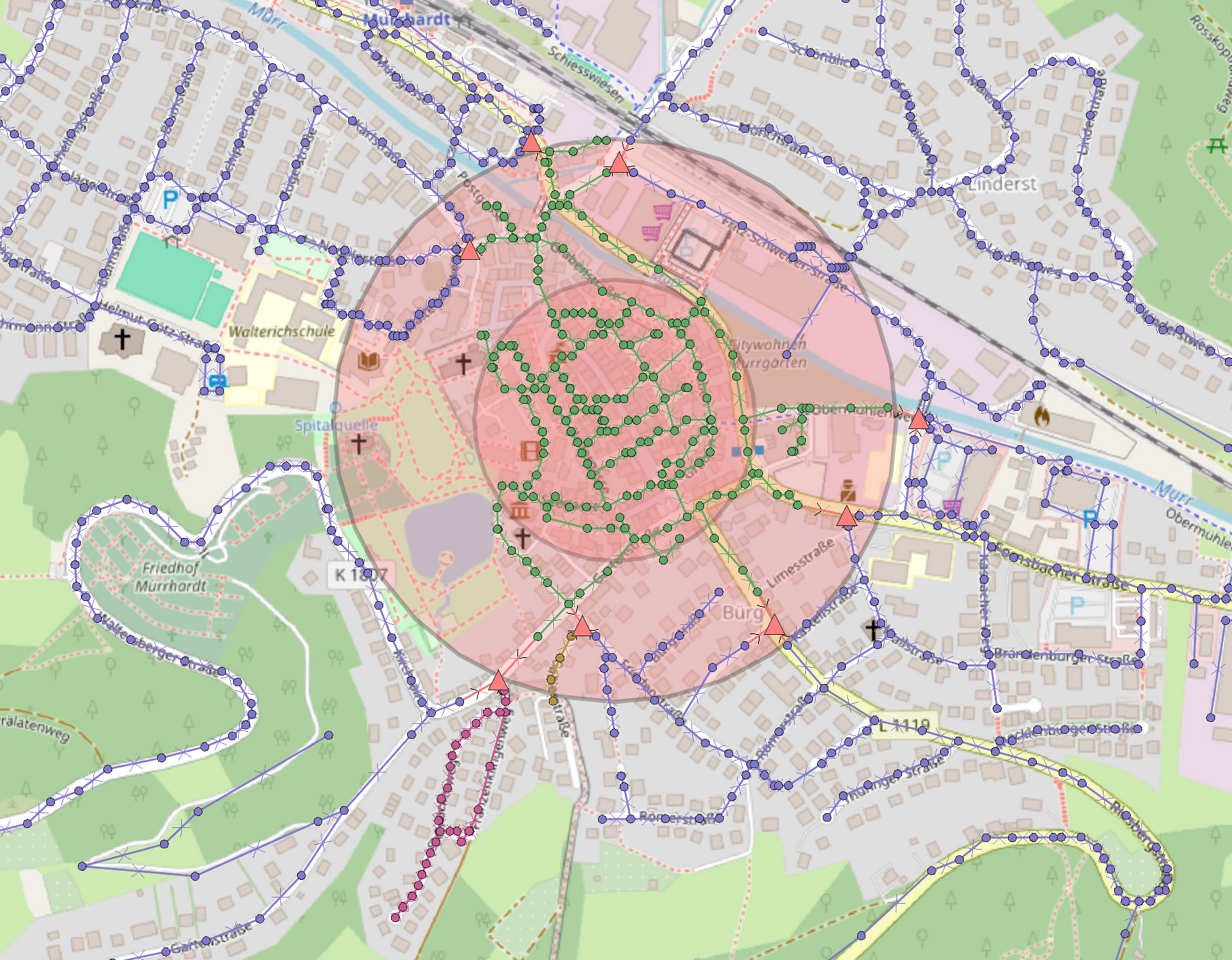}};
        \end{tikzpicture}
        \caption{Global view of the small graph and its separator in the top left corner, overlaid with a zoom on the separator, where one can also observe smaller components and balls.
        The largest component is of size \num{98870}, which indeed gives a balance of $0.997$.}
        \label{fig:illus_sep_small_original_bb}
    \end{subfigure}\hfill
    \begin{subfigure}[t]{0.48\textwidth}
        \centering
        \includegraphics[trim={7.6cm 0cm 7.6cm 0cm},clip,width=\textwidth]{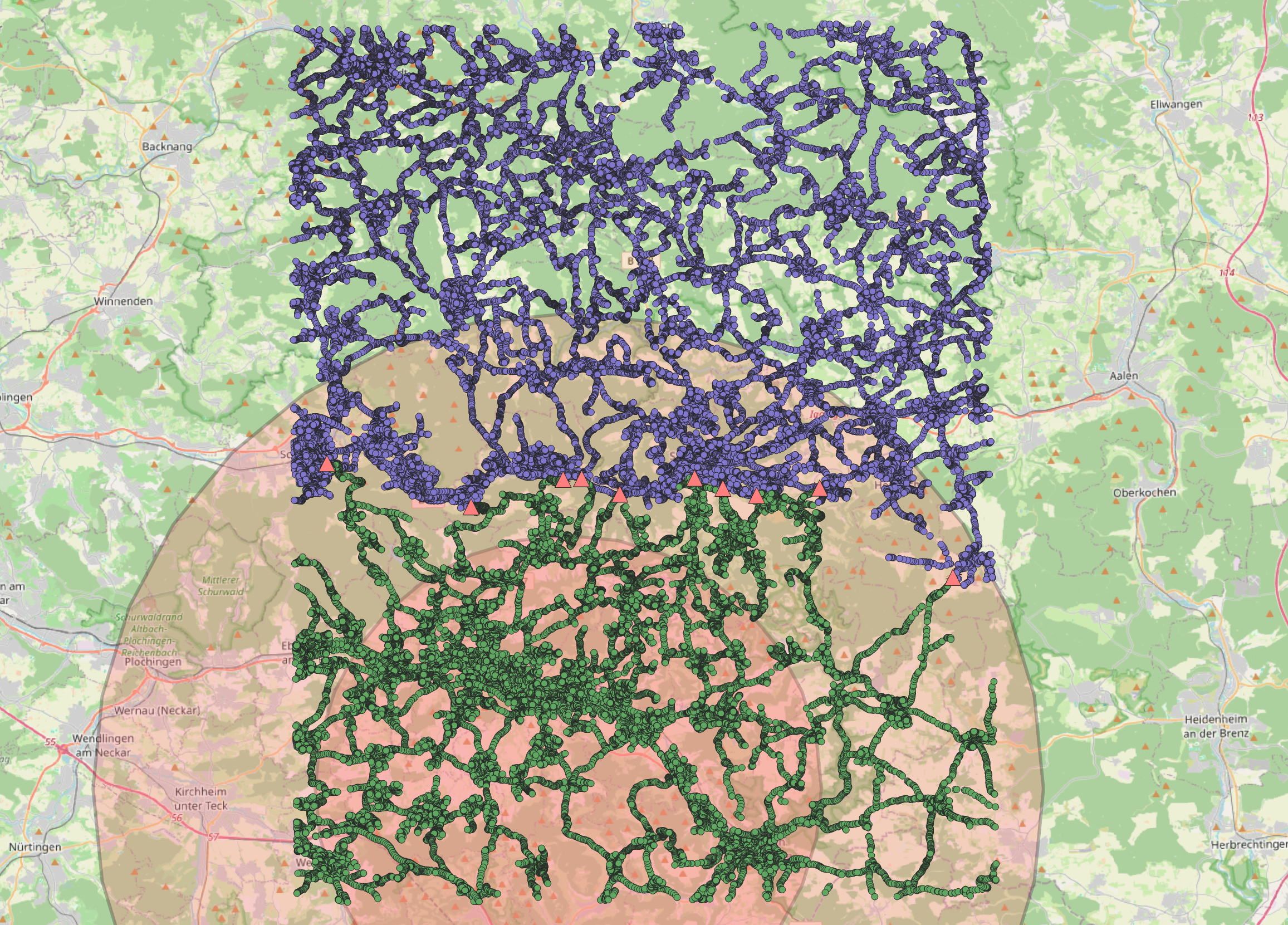}
        \caption{The resulting balance factor is $0.546$ with a very small separator of size 10. \emph{KaHIP} identifies almost the same separator in many runs.}
        \label{fig:illus_sep_small_good}
    \end{subfigure}
    \caption{Separation of the small graph using different ball strategies.
    On the left is the original one with \emph{Better Balance}.
    On the right is an output of the engineered sampling strategy, with $100$ vertices sampled and \SI{50}{\percent} probability of separating.}
    \label{fig:illus_sep_small}
\end{figure}

Instead a more practical but randomized ball strategy consists in sampling the
vertices, finding the exact $k$- and $(1-k)$-enclosing balls $B = b(p,r)$ and $B'
= b(p,R)$ from each vertex point $p$, and keeping the best ball pair.
This method also allows for arbitrary balance.
One could define the best pair to be the one
that guarantees the best theoretical upper bound $\frac{a}{a-1} c$ on the separator size,
maximizing the radius ratio $a = \frac{R}{r}$, or also
randomly run separations using intermediately computed balls and keep the one
yielding the smallest separator size.
To prevent this ratio $a$ from diverging, one could set a threshold and fall back to nested squares ball strategy if no suitable ball is found.
In practice, setting a threshold of $1.2$ still ensures a maximum separator size of $\frac{a}{a-1}c = 6c$, and such balls can always be found on the small graph; see \cref{fig:select_sep_Rr_relation}.
\begin{figure}[ht]
    \centering
    \includegraphics[width=0.65\textwidth]{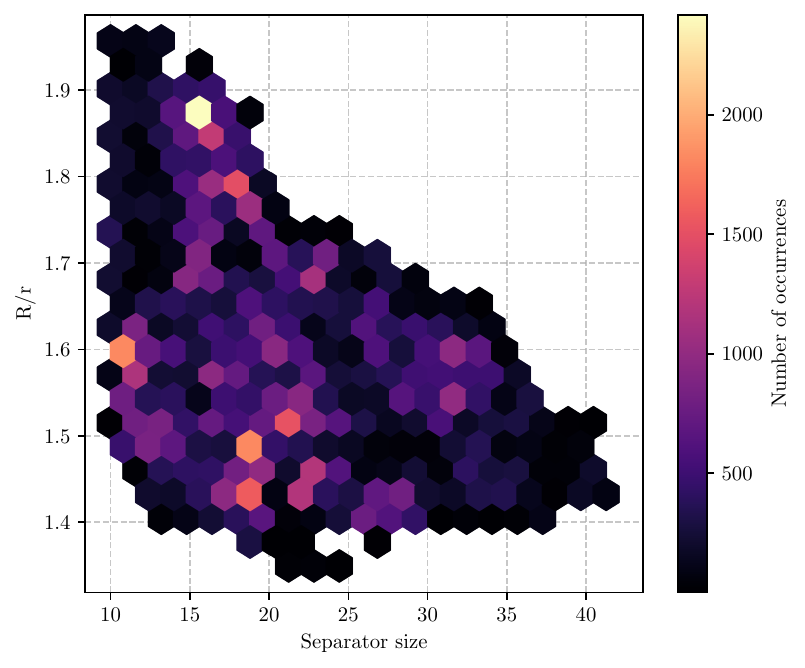}
    \caption{Density of separators w.r.t. ball radius ratio $a=\frac{R}{r}$ and separator size. This is done on the small graph, taking each vertex as a ball centre, and targeting a $\frac{2}{3}$ balance.}%
    \label{fig:select_sep_Rr_relation}
\end{figure}
This figure also illustrates the fact that maximizing $a$ induces smaller
separator sizes.
It also makes sense, because there should be more paths in the ring for the
max-flow algorithm to consider.
Consequently, during the sample strategy, one could accumulate and sort balls by
decreasing radius ratio $a$, and run separations by processing balls in order.
Combining both approaches, we then run the separation using the best theoretical pair of balls and keep the smallest separator.
Radii $r$ or $R$ can be found naively by computing distances from a sample
vertex $p$ to each vertex point. Selecting the $k$th farthest vertex point can
be done in $\O(n)$ thanks to a selection algorithm such as introselect. The ball
radius is set to the distance to this $k$th farthest vertex point.
This process can be greatly optimized by precomputing a spanning grid of $\log^2(n)$ cells
which holds their covered vertices. The query to find distance $r$ from $p$
consists in accumulating cells by exploring the grid outwards from the cell of
$p$. When at least $k$ vertices are accumulated, one only have to compute
distances to the last expanded annulus to determine the exact $k$th-farthest
vertex. See \cref{fig:grid-illus} for an illustration of this process.
To accelerate the process even further, distance computations can be
skipped at the cost of an approximation on $r$.
\begin{figure}[ht]
        \begin{subfigure}[t]{0.32\textwidth}
        \centering
        \includegraphics[page=10,width=\textwidth]{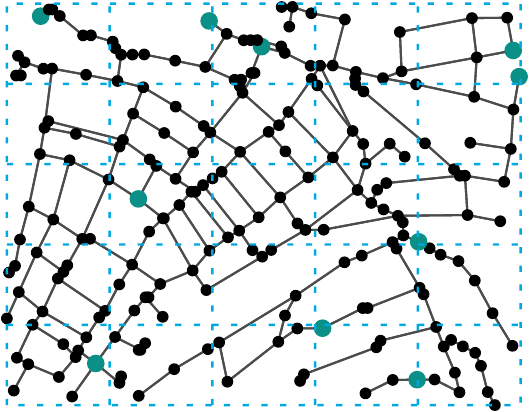}
        \caption{Pre-computation of the {\color{kn_seeblau}spanning grid} of $\log^2(n)$ cells. For the sake of clarity, more cells are used in this illustration.
        {\color{kn_seegruen}Sampled vertices} are highlighted.}
        \label{fig:grid-illus-step-grid}
    \end{subfigure}\hfill
    \begin{subfigure}[t]{0.32\textwidth}
        \centering
        \includegraphics[page=12,width=\textwidth]{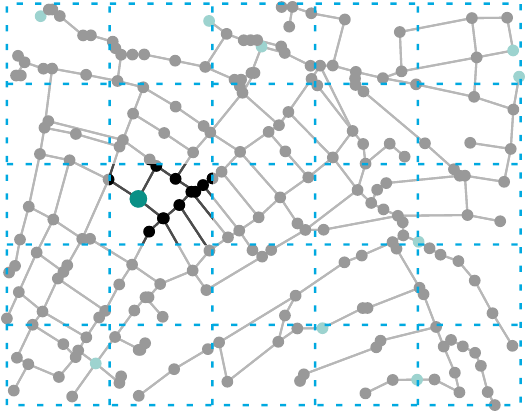}
        \caption{First step of the grid exploration from the cell containing {\color{kn_seegruen}one sampled vertex $p$}, in order to find the radius $r$ of the inner ball $B$ centred at ${\color{kn_seegruen}p}$.}
        \label{fig:grid-illus-step-explo1}
    \end{subfigure}\hfill
    \begin{subfigure}[t]{0.32\textwidth}
        \centering
        \includegraphics[page=13,width=\textwidth]{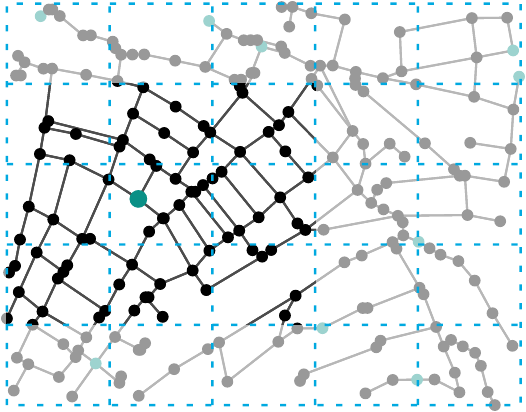}
        \caption{Second step of the grid exploration, accumulating vertices in cells of the next ring. It is also the last expansion, supposing enough vertices ($(1-\alpha)n$) were accumulated.}
        \label{fig:grid-illus-step-explo2}
    \end{subfigure}
    \begin{subfigure}[t]{0.32\textwidth}
        \centering
        \includegraphics[page=14,width=\textwidth]{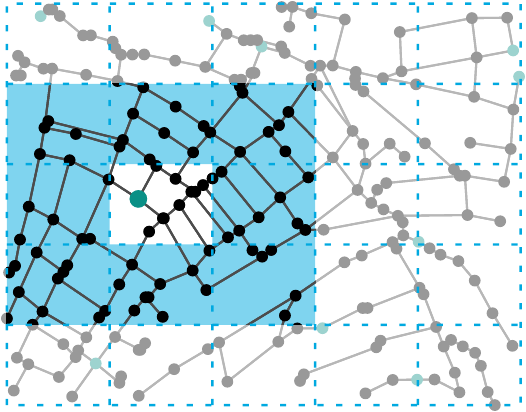}
        \caption{Distances from ${\color{kn_seegruen}p}$ to vertices in the {\color{kn_seeblau}last expanded ring} are computed.}
        \label{fig:grid-illus-step-dist-comp}
    \end{subfigure}\hfill
    \begin{subfigure}[t]{0.32\textwidth}
        \centering
        \includegraphics[page=15,width=\textwidth]{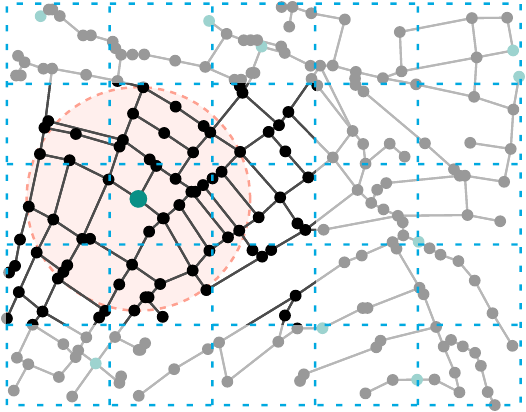}
        \caption{From previously computed distances, the exact $(1-\alpha)n$\textsuperscript{th}-farthest vertex is identified, and hence $r$, giving the complete inner ball ${\color{kn_peach}B({\color{kn_seegruen}p},r)}$.}
        \label{fig:grid-illus-step-inner-ball}
    \end{subfigure}\hfill
    \begin{subfigure}[t]{0.32\textwidth}
        \centering
        \includegraphics[page=16,width=\textwidth]{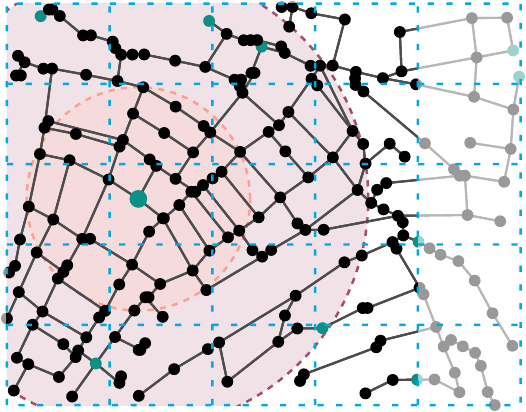}
        \caption{The outer ball ${\color{kn_bordeaux65}B' = B({\color{kn_seegruen}p},R)}$ is found in the same way. The grid exploration is resumed until $\alpha n$ vertices are accumulated, distances to the last expanded ring are computed, and $R$ is deduced.}
        \label{fig:grid-illus-step-outer-ball}
    \end{subfigure}
    \caption{Illustration of the \emph{sampling} ball strategy using a precomputed spanning grid.}
    \label{fig:grid-illus}
\end{figure}
To determine the necessary parameters, namely sample size and expected number of separation runs, we benchmarked several configurations on a representative road network.
According to results shown \cref{fig:benchmark_sampling_smalllest_sep}, a sample of size 100 with 50 separation runs seems to yield the smallest separator for  a balance factor of  $\alpha=\frac23$ almost every time. The smallest separator is also found in the 50 first tries \SI{75}{\percent} of the time with larger separation runs; see \cref{fig:benchmark_sampling_rank_smallest_sep}.
\begin{figure}[ht]
    \begin{subfigure}[t]{\textwidth}
        \centering
        \includegraphics[width=0.48\textwidth]{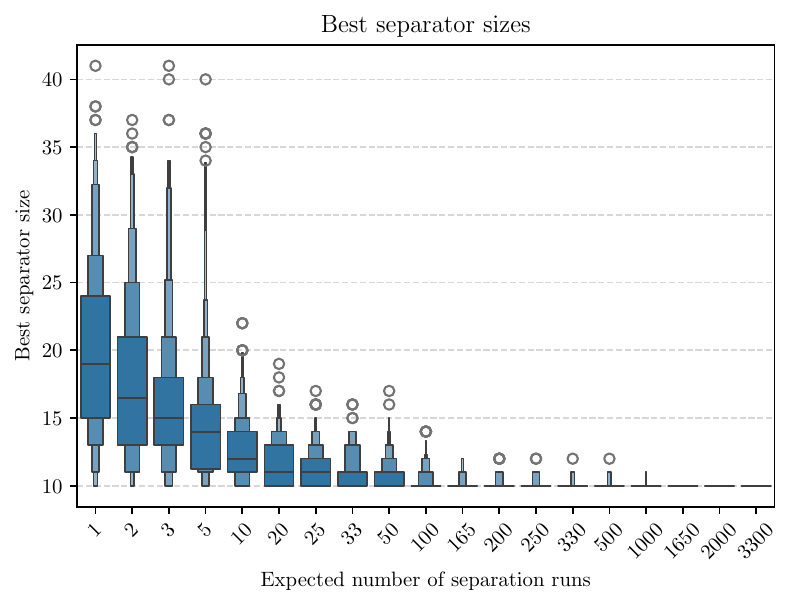}
        \includegraphics[width=0.48\textwidth]{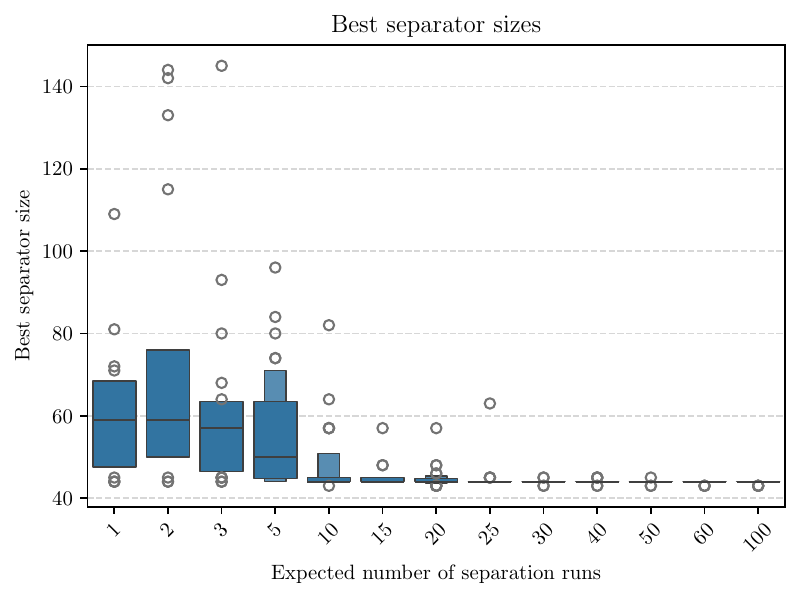}
        \caption{Smallest separator sizes w.r.t. number of samples and expected number of separation runs.}
        \label{fig:benchmark_sampling_smalllest_sep}
    \end{subfigure}
    \begin{subfigure}[t]{\textwidth}
        \centering
        \includegraphics[width=0.48\textwidth]{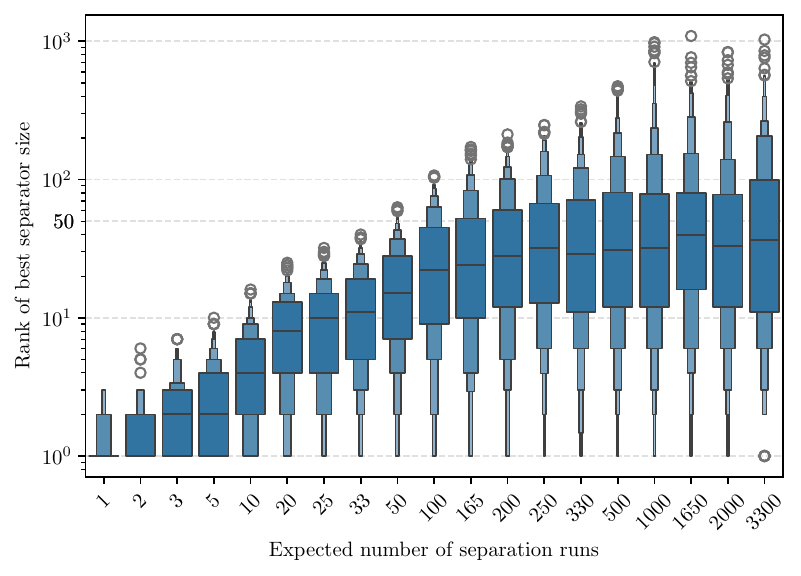}
        \includegraphics[width=0.48\textwidth]{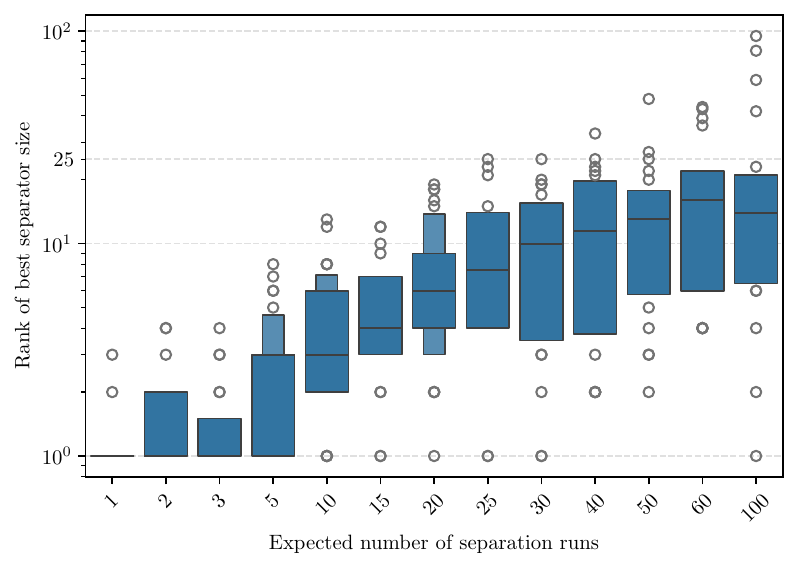}
        \caption{Rank of smallest separator w.r.t. expected number of separation runs.}
        \label{fig:benchmark_sampling_rank_smallest_sep}
    \end{subfigure}
    \caption{Benchmark results to determine best sampling parameters for the corresponding ball strategy. Left is the small instance, and right is the large one.}
\end{figure}

In order to assess the  practical performance of this separation algorithm, we compare its performance against the heavy-engineered KaHIP \cite{sanders2013think,sanders2016advanced}.
To make it a fair fight, we first separate the graph using our algorithm and measure the balance factor. Then, we separate using KaHIP by asking two components and giving the same balance as target. Its mode is also set to \texttt{FAST} or \texttt{STRONG}, which enables a compromise between separation time and size.
On the small graph, our algorithm yields better separator sizes while using less time at 10 to 35 separation runs when compared to KaHIP-\texttt{FAST}, and at 30 to 110 separation runs when compared to KaHIP-\texttt{STRONG}; see \cref{fig:select_ball_strategy_benchmarks}.
On the large graph, up to 100 separation runs performed during the ball strategy execution, the separation is either slower but of higher quality against KaHIP-\texttt{FAST}, or faster while yielding roughly the same size against KaHIP-\texttt{STRONG}.
In the context of nested dissection, only a very few amounts of separation are done on large subgraphs, though, and most of the separations are performed on smaller graphs, closer to the previous, small one.
\begin{figure}[!ht]
    \includegraphics[width=0.49\textwidth]{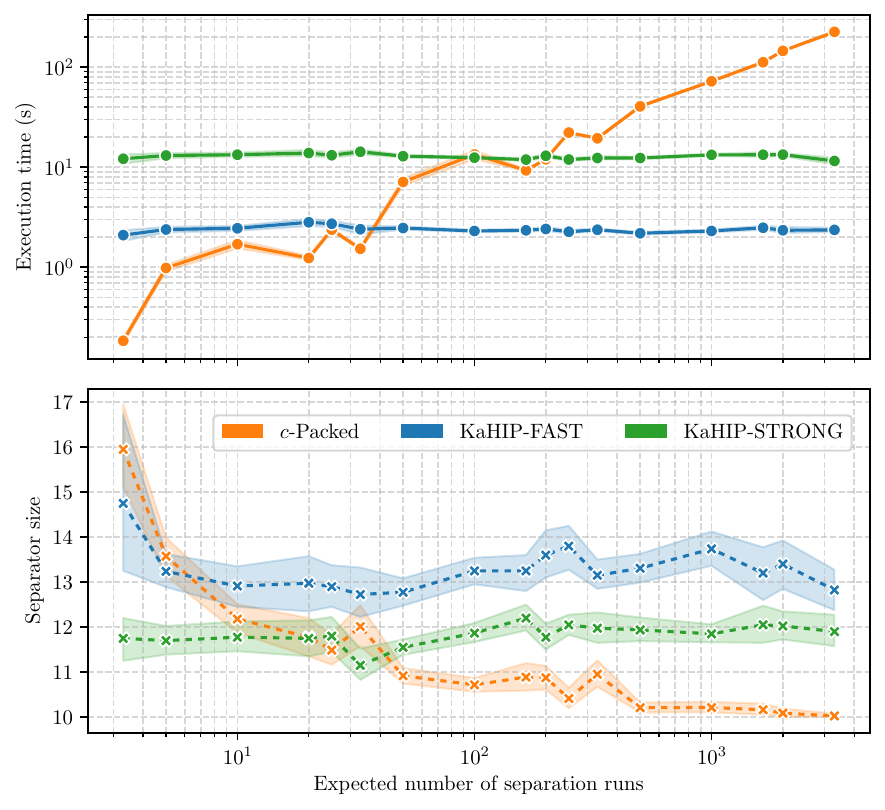}\hfill
    \includegraphics[width=0.49\textwidth]{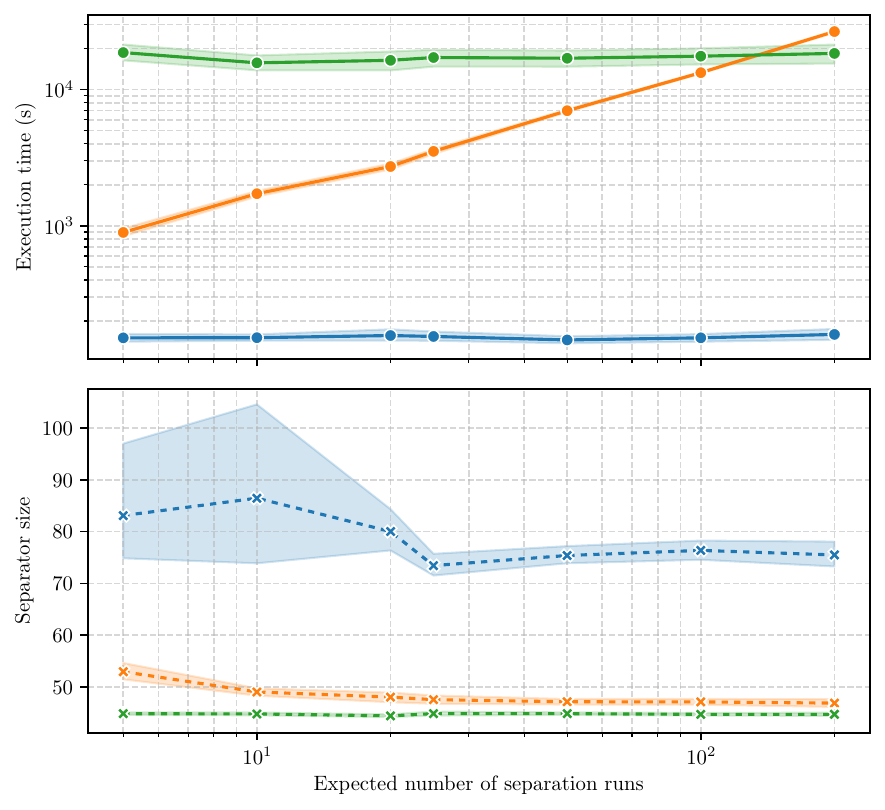}
    \caption{Time and separator size comparison between our separation algorithm and KaHIP in \texttt{FAST} and \texttt{STRONG} modes, w.r.t.\ the expected number of separation runs, for the small dataset (left) and the large dataset (right).
    The $c$-packed balance factor in the ball strategy is always set to $\frac23$.
    For the large dataset, two KaHIP separator outliers (larger than \num{1000}) were removed, one was of size \num{54000}. This again shows the importance of having a theoretical guarantee on the separator size.}
    \label{fig:select_ball_strategy_benchmarks}
\end{figure}

\subsection{Exact Distance Oracles}
The separator decomposition was implemented just as described \cref{sec:edo}.
using our separator algorithm described above as a subroutine.
The tree decomposition was fully implemented as described in \cref{subsec:tree_dec_cpacked}. A small optimization
was added: instead of recurring on $\bigl(G[C \cup S],(C \cap W) \cup S\bigr)$, we extract
the set $N(C) \cap S$---vertices of $S$ connected to
$C$---in time $\O(|S| \Delta)$ and use it in place of $S$.
This optimization resulted in an average bag size smaller by a factor of $\frac13$. 
Both structures have been constructed on the large dataset. Regarding the pre-computation, they both took around \SI{4.5}{\hour} to compute, without any parallelization.
The computed separator decomposition has a height of 24 with \num{1700000} bags, whereas the tree decomposition is 19 nodes high and is made of \num{1200000} bags. 
As discussed earlier \cref{sec:edo}, those two data structures can be used to define Exact Distance Oracles, namely $\mathcal{S}$-EDO and $\mathcal{T}$-EDO.
In order to finally compare them, we investigate the number of distance comparisons done at query phase---which is expected to be smaller by a factor of $\O(\log n)$ for $\mathcal{T}$-EDO. The query time should be proportional.
We can observe this only from the separation and tree decompositions, based on bag sizes.
Results indeed show a clear domination of the $\mathcal{T}$-EDO, as show on \cref{fig:sedo_vs_tedo_query}. On average, it requires only 20 distance comparisons to answer a query, whereas the $\mathcal{S}$-EDO needs a bit less than 4 times more. The \SI{75}{\percent} quantile of $\mathcal{T}$-EDO ($14$) is below the \SI{25}{\percent} quantile of $\mathcal{S}$-EDO ($43$).
\begin{figure}[!ht]
    \centering
    \includegraphics[width=0.7\linewidth]{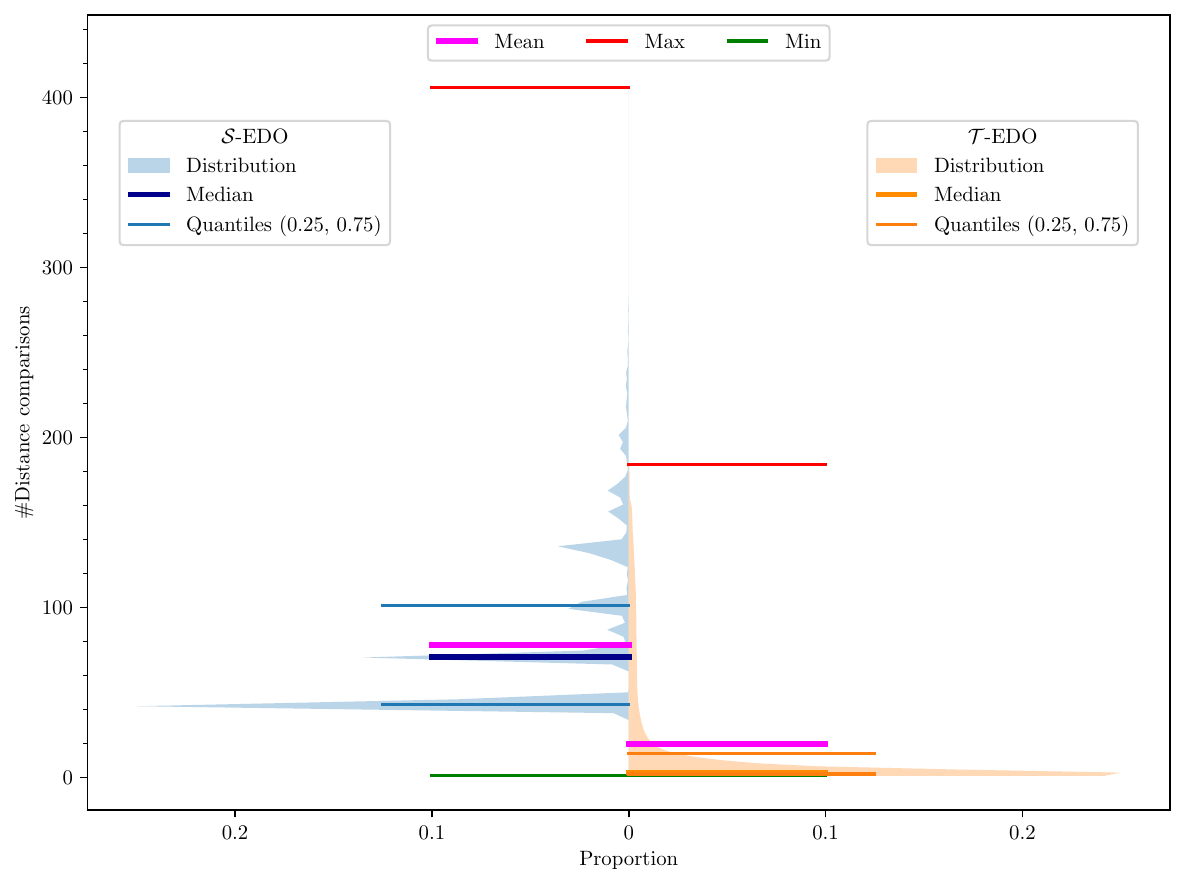}
    \caption{Distribution of the number of distance comparisons required for a query, on average, between $\mathcal{S}$-EDO and $\mathcal{T}$-EDO.}
    \label{fig:sedo_vs_tedo_query}
\end{figure}
The tree decomposition can be used as an intermediate step to solve many more problems.
Its very small average bag size might even make algorithms with exponential dependency on $\tw$ practical.


\section{Conclusions and Future Work}
Our experiments reveal that $\lambda$ and $c$ are small enough in large real-world road networks for parameterized bounds to be meaningful. Especially $\lambda$ values turned  out to be really low, which motivates further research into algorithms parameterized by it. 
We also show that such parameterized algorithms can be engineered to perform very well in practice while upholding theoretical guarantees. In particular, our exact distance oracle with $\mathcal{O}(c)$ query time requires only few  distance computations on average. Indeed, even the maximum oracle time is much better than the computed $c$-value which might warrant further investigation.

Balanced separator computation is at the core of our implemented algorithms.
Very recently, Har-Peled \cite{harpeled2026separator} presented a simple randomized
algorithm that computes balanced separators of size $\mathcal{O}(c)$ in expected
$\mathcal{O}(cn)$ time. The algorithm is closely related to separator
constructions for $\tau$-lanky and $\lambda$-low-density graphs
\cite{le2024greedy} and avoids the max-flow computations that currently dominate
the running time of our implementation. Consequently, it has the potential to
significantly accelerate our algorithms, although it remains to be seen whether
it can match the separator quality achieved by our approach in practice.
Incorporating Har-Peled's algorithm into our tree decomposition algorithm
reduces its running time to $\mathcal{O}(cn \log n)$. The preprocessing time of
the exact distance oracle, however, would still be dominated by the subsequent
distance computations, which require $\mathcal{O}(c^3 n \log n)$ time. Reducing
this dependence on $c$ remains an interesting open problem.

\clearpage
\bibliography{refs}

\newpage
\appendix
\section{Improved Approximation Algorithm for Smallest Enclosing Square}%
\label{apx:enclosing-square}
Given a set of $n$ points in the plane and a parameter $k$, the goal is to compute an axes-parallel square with minimum side length  that encloses $k$ points. Smid \cite{smid1992finding} presented an algorithm for this problem with a running time in $\mathcal{O}(n \log n +  n k \log^2 k)$ and a linear space consumption. Datta et al.~\cite{datta1995static} reduced the running time to $\mathcal{O}(n \log n +  n \log^2 k)$. Chan \cite{chan1998geometric} described a randomized algorithm that runs in $\mathcal{O}(n \log n)$ for arbitrary $k$ and Mahapatra \cite{mahapatra2011k} a deterministic one with a running time in $\mathcal{O}(n \log^2 n)$.   Using  this exact algorithm directly for  our separator computation  in $c$-packed graphs yields a separator size of~$|S_{0.8}| = 2 \sqrt 2 c$, and accordingly we get $|S_{2/3}|= 4 \sqrt 2 c$ in $\mathcal{O}(c^2 n + n \log ^2 n)$ time.

However, for the purpose of more efficient separator computation, we strive for a linear running time but are also content with a $(1+ \varepsilon)$-approximation, that is, we can accept a slightly larger than optimal square. For disks this was shown to be possible in \cite{harpeled2005fast}, where an algorithm with a running time in  $\O(n + \frac1{\varepsilon^3}
\log^2 \frac1\varepsilon)$ for $\varepsilon > 0$ and $k \in \Omega(n)$ was presented. In the following, we discuss that this algorithm can be applied to axes-parallel squares with only minor modifications. As squares are arguably even simpler to treat, we decrease the running time to $\O(n + \frac1{\varepsilon^2}
\log \frac1\varepsilon)$ in the process.

The algorithm works as follows: First, it computes a $2$-approximation of the
optimal radius $r_{OPT}$ in linear time. This algorithm does not rely on disk
geometry and  directly transfers to squares. The returned radius $r \leq 2
r_{OPT}$ is used to construct a grid. It is argued that each grid cell can be
covered by at most 5 disks of radius $r/2$and thus the maximum number of points
contained in each grid cell is upper bounded by $5k$. For squares, this bound
reduces to $4k$. Clearly, the optimal disk/square can not extend beyond a  $3
\times 3$ grid cell cluster, as the radius would then exceed $2r > 2r_{OPT}$.
Therefore, it suffices to inspect each such cluster which contains at least $k$
points. The current cluster is further subdivided using a grid of side length
$\frac{1}{\varepsilon}r$ and the cluster points are snapped to their closest
grid points. Thus, we go from at most $45k$ distinct points in the cluster
($36k$ for squares) to $\mathcal{O}\left(\frac{1}{\varepsilon^2}\right)$ snapped
points. We know strive the find the smallest $k$-enclosing disk/square for the
snapped points, respecting the multiplicities of points with the same
coordinate. For squares, it is easy to see that we can restrict at least one
pair of diagonally opposing corners do be grid points. Thus for each grid point,
we assume it is w.l.o.g. the lower left corner of the square (and then repeat
the process for lower right, upper left and upper right corner). We then want to
find the smallest side length such that the resulting square contains at least
$k$ points. We do so via binary search over the possible side lengths. To
quickly answer how many points are contained in such a square, we use a
2-dimensional prefix sum array which can be precomputed for the grid point
counts in linear time and which returns the correct value in $\mathcal{O}(1)$.
Thus, for each grid point and corner role, we need $\mathcal{O}\left(\log
\frac{1}{\varepsilon}\right)$ time and
$\mathcal{O}\left(\frac{1}{\varepsilon^2}\log \frac{1}{\varepsilon}\right)$ in
total. For disks, the data structures that report the point counts are more
complicated and thus the dependency on $\varepsilon$ is worse.
Returning the smallest square seen in this process is a valid
$(1+\varepsilon)$-approximation, as we only need to enlarge the square by at
most this factor to enclose all relevant original points to fit the bound of
$k$. All other calculations from \cite{harpeled2005fast} regarding the running
time analysis (including how often clusters with many points occur) are
oblivious to the disk geometry and carry over unchanged to squares.
Therefore, we get the following theorem.

\begin{theorem}
    Given a set of points in the plane, $k \in \Theta(n)$, and $\epsilon>0$, we can compute a
    $(1+\varepsilon)$-approximation of the smallest axis-aligned $k$-enclosing
    square in time $\O(n + \frac1{\varepsilon^2} \log \frac1\varepsilon)$.
\end{theorem}

\end{document}